\documentclass{article}
\usepackage{ijcai23}

\usepackage{times}
\usepackage{soul}
\usepackage{url}
\usepackage[hidelinks]{hyperref}
\usepackage[utf8]{inputenc}
\usepackage[small]{caption}
\usepackage{graphicx}
\usepackage{booktabs}
\usepackage[switch]{lineno}

\title{First-Choice Maximality Meets Ex-ante and Ex-post Fairness
}

\author{
Xiaoxi Guo$^1$
\and
Sujoy Sikdar$^2$\and
Lirong Xia$^{3}$\and
Yongzhi Cao$^{1*}$ \And
Hanpin Wang$^{4,1}$\\
\affiliations
$^1$Key Laboratory of High Confidence Software Technologies
        (MOE),\\ School of Computer Science, Peking University\\
$^2$Department of Computer Science, Binghamton University\\
$^3$Department of Computer Science, Rensselaer Polytechnic Institute\\
$^4$School of Computer Science and Cyber Engineering, Guangzhou University\\
\emails
guoxiaoxi@pku.edu.cn,
ssikdar@binghamton.edu,
xialirong@gmail.com,
\{caoyz, whpxhy\}@pku.edu.cn
}

\usepackage{amsmath}
\usepackage{amsthm}
\usepackage{amsfonts}
\usepackage{amssymb}
\usepackage{algorithm}
\usepackage[noend]{algpseudocode}
\usepackage{algorithmicx}
\usepackage{booktabs}
\usepackage[capitalise,noabbrev]{cleveref}
\usepackage{color}
\usepackage{diagbox}
\usepackage[inline]{enumitem}
\usepackage{etoolbox}
\usepackage{mathtools}
\usepackage{multirow}
\usepackage{soul}
\usepackage{tabularx}
\usepackage{thmtools}
\usepackage[usenames,svgnames,table]{xcolor}\usepackage{thm-restate}
\usepackage{comment}

\algrenewcommand\algorithmicindent{0.75em}

\makeatletter
\newcommand{\algmargin}{\the\ALG@thistlm}
\makeatother
\newlength{\whilewidth}
\algdef{SE}[parWHILE]{parWhile}{EndparWhile}[1]
{\parbox[t]{\dimexpr\linewidth-\algmargin}{%
		\hangindent\whilewidth\strut\algorithmicwhile\ #1\ \algorithmicdo\strut}}{\algorithmicend\ \algorithmicwhile}%
\algnewcommand{\parState}[1]{\State%
	\parbox[t]{\dimexpr\linewidth-\algmargin}{\strut #1\strut}}

\usepackage{tikz}
\usetikzlibrary{fit,calc}

\colorlet{mygray}{gray!40}

\newcommand*\bcircled[1]{\tikz[baseline=(char.base)]{
		\node[shape=circle,draw=black,inner sep=1pt] (char) {#1};}}

\crefname{prop}{Proposition}{Propositions}
\crefname{thm}{Theorem}{Theorems}
\crefname{lem}{Lemma}{Lemmas}

\newcommand{\impord}[1]{\vartriangleright_{#1}} 

\newcommand{\ma}{\mathcal{A}}

\newcommand{\mr}{\mathcal{R}}
\newcommand{\mP}{\mathcal{P}}

\newcommand{\ram}{random} 

\newcommand{\sd}[1]{\nobreak\succeq^{sd}_{#1}\allowbreak}

\newcommand{\fsdef}{\text{sd-}\allowbreak\text{envy-}\allowbreak\text{freeness}}

\newcommand{\sdef}{\text{sd-}\allowbreak\text{EF}}

\newcommand{\fsdopt}{\text{sd-}\allowbreak\text{efficiency}}

\newcommand{\sdopt}{\text{sd-}\allowbreak\text{E}}

\newcommand{\fsdsp}{\text{sd-}\allowbreak\text{weak-}\allowbreak\text{strategy}\allowbreak\text{proofness}}

\newcommand{\sdsp}{\text{sd-}\allowbreak\text{WSP}}

\newcommand{\ucs}{U}			

\newcommand{\fsdwef}{\text{sd-}\allowbreak\text{weak-}\allowbreak\text{envy-}\allowbreak\text{freeness}}

\newcommand{\sdwef}{\text{sd-}\allowbreak\text{WEF}}

\newcommand{\ld}[1]{\succ_{#1}^{lexi}}

\newtheorem{definition}{Definition}
\newtheorem{example}{Example}

\newtheorem{remark}{Remark}

\newcommand{\am}{\text{EBM}}

\newcommand{\rk}[2]{rk(#1,#2)}

\newcommand{\ffhr}{\text{favoring-}\allowbreak\text{higher-}\allowbreak\text{ranks}}
\newcommand{\Ffhr}{\text{Favoring-}\allowbreak\text{higher-}\allowbreak\text{ranks}}
\newcommand{\fhr}{\text{FHR}}

\newcommand{\ffhcr}{\text{favoring-}\allowbreak\text{eagerness-}\allowbreak\text{for-}\allowbreak\text{remaining-}\allowbreak\text{items}}
\newcommand{\fhcr}{\text{FERI}}

\newcommand{\fopt}{\text{Pareto-}\allowbreak\text{efficiency}}

\newcommand{\opt}{\text{PE}}

\newcommand{\citeay}[1]{\citeauthor{#1} [\citeyear{#1}]}

\newcommand{\tp}[2]{\nobreak top(#1,#2)\allowbreak}

\newcommand{\frkm}{\text{rank-}\allowbreak\text{maximality}}

\newcommand{\tps}[2]{T_{#1,#2}}

\newcommand{\rka}[3]{rk(#1,#2,#3)}

\newcommand{\fefo}{\text{envy-}\allowbreak\text{free }\allowbreak{up }\allowbreak{to }\allowbreak{one }\allowbreak{item}}
\newcommand{\efo}{EF1}
\newcommand{\fefx}{\text{envy-}\allowbreak\text{free }\allowbreak{up }\allowbreak{to }\allowbreak{any }\allowbreak{item}}

\newcommand{\fsdefo}{\fefo}
\newcommand{\sdefo}{\efo}

\newcommand{\bm}{\text{BM}}

\newcommand{\mam}{\text{GEBM}}
\newcommand{\fmam}{\text{generalized }\allowbreak{eager}\allowbreak{ Boston }\allowbreak{mechanism}}
\newcommand{\mpbm}{\text{GPBM}}
\newcommand{\fmpbm}{\text{generalized }\allowbreak{probabilistic }\allowbreak{Boston mechanism}}

\newcommand{\ntr}{neutrality}

\newcommand{\fcm}{\text{FCM}}

\newcommand{\ffcm}{\text{first-}\allowbreak\text{choice }\allowbreak\text{maximality}}
\newcommand{\Ffcm}{\text{First-}\allowbreak\text{choice }\allowbreak\text{maximality}}

\newcommand{\Y}{\textcolor{black}{Y}}
\newcommand{\N}{\textcolor{black}{N}}

\colorlet{mycolor}{blue!20}

\makeatletter
\newcommand{\multiline}[1]{%
  \begin{tabularx}{\dimexpr\linewidth-\ALG@thistlm}[t]{@{}X@{}}
    #1
  \end{tabularx}
}
\makeatother

\newcommand{\dal}{\text{deterministic }\allowbreak\text{allocation}}
\newcommand{\das}{\text{deterministic }\allowbreak\text{assignment}}

\newcommand{\dals}{\text{deterministic }\allowbreak\text{allocations}}
\newcommand{\dass}{\text{deterministic }\allowbreak\text{assignments}}

\renewcommand{\dal}{\text{allocation}}
\renewcommand{\das}{\text{assignment}}

\renewcommand{\dals}{\text{allocations}}
\renewcommand{\dass}{\text{assignments}}

\begin{document}

\maketitle

\begin{abstract}
For the assignment problem where multiple indivisible items are allocated to a group of agents given their ordinal preferences, we design randomized mechanisms that satisfy {\em \ffcm{}} (\fcm{}), i.e., maximizing the number of agents assigned their first choices, together with {\em \fopt{}} (\opt{}).
Our mechanisms also provide guarantees of ex-ante and ex-post fairness. 
The {\em \fmam{}} is ex-ante envy-free, and ex-post {\em envy-free up to one item} (\sdefo{}). 
The {\em \fmpbm{}} is also  ex-post \sdefo{}, and satisfies ex-ante efficiency instead of fairness.
We also show that no strategyproof mechanism satisfies ex-post \opt{}, \sdefo{}, and \fcm{} simultaneously.
In doing so, we expand the frontiers of simultaneously providing efficiency and both ex-ante and ex-post fairness guarantees for the assignment problem.
\end{abstract}

\section{Introduction}

How should $m$ indivisible items be assigned to $n$ agents efficiently and fairly when the agents have heterogeneous preferences over the items? This {\em assignment problem} is fundamental to economics, and increasingly, computer science, due to its versatility in modeling a wide variety of real-world problems such as assigning computing resources in cloud computing~\cite{Ghodsi11:Dominant,Grandl15:Multi}, courses to students in colleges~\cite{budish2011combinatorial}, papers to referees~\cite{garg2010assigning}, and medical resources in healthcare~\cite{kirkpatrick2020scarce,pathak2021fair,aziz2021efficient} where agents may obtain multiple items.
In these problems, efficiency and fairness are the common desiderata, and also the goal of mechanism design.

Efficiency reflects the degree of agents' satisfaction and the room of improvement for the given assignment.
The number of agents who are allocated their first choices or one of their top $k$ choices is often highlighted as a measure of efficiency in practice~\cite{chen2006school,li2020ties,irving2006rank}.
{\em \Ffcm{}} (\fcm{}), i.e., maximizing the number of agents allocated their first choices~\cite{dur2018first}, is considered to be either highly desirable or indispensable in problems like job markets~\cite{kawase2020subgame}, refugee reallocation~\cite{sayedahmed2022centralized}, and school choice~\cite{friedman1955role,friedman1962capitalism}.
In addition, {\em Pareto efficiency} (\opt{}) is often considered a basic efficiency property, which requires that the items cannot be redistributed in a manner strictly preferred by some agents and no worse for every other agent.
In other words, it urges that all the improvements without undermining agents' benefits should be made.

Guaranteeing the fairness of assignments is also an important consideration, and {\em envy-freeness} (EF)~\cite{gamow1958puzzle,foley1967resource}, which requires that no agent prefers the allocation of another agent to her own, is an exemplar of fairness requirements. However, an envy-free assignment may not exist for indivisible items, for example, when agents have identical preferences. Exact envy-freeness can only be guaranteed by randomization over assignments. This allows an ex-ante guarantee that every agent values its expected allocation, interpreted as probabilistic ``shares'' of items, at least as much as that of any other agent when using the notion of {\em stochastic dominance} (sd)~\cite{Bogomolnaia01:New} as the method of comparison. Such a random assignment describes a probability distribution over all the assignments.

Envy-freeness can also be achieved approximately ex-post.
{\em Envy-freeness up to one item} (\efo{}), which guarantees that any pairwise envy among agents is eliminated by removing one item from the envied agent's allocation~\cite{budish2011combinatorial}, is popular among such approximations for its compatibility with efficiency~\cite{caragiannis2019unreasonable}.

\begin{table*}[htp!]
\caption{Comparison of the properties guaranteed by RSDQ, PS-Lottery, \mpbm{} and \mam{}.}\label{tab:properties}
\begin{center}
\begin{tabular}{|c|cc|c|c|cc|c|}
    \hline

    \multirow{3}{*}{Mechanism} & \multicolumn{3}{c|}{Efficiency} & \multicolumn{3}{c|}{Fairness} & \multirow{2}{*}{Strategyproofness}\\\cline{2-7}
    & \multicolumn{2}{c|}{ex-post} & ex-ante & ex-post & \multicolumn{2}{c|}{ex-ante} & \\\cline{2-8}

    & \fcm{} 
    & \opt{} & \sdopt{} & \sdefo{} & \sdwef{} & \sdef{} & \sdsp{} 
    \\\hline
    
    RSDQ
    & \N$^\texttt{a~}$ 
    & \Y$^\texttt{b~}$ & \N$^\texttt{c~}$
    & \N$^{P\ref{prop:rsdpnefo}}$ & \Y$^\texttt{b~}$ & \N$^\texttt{c~}$
    & \Y$^\texttt{b~}$ \\
    
    PS-Lottery
    & \N$^\texttt{a~}$
    & \Y$^\texttt{d~}$ & \Y$^\texttt{d~}$
    & \Y$^\texttt{d~}$ & \Y$^\texttt{d~}$ & \Y$^\texttt{d~}$ 
    & ~\N$^\texttt{e~~}$ \\
    
    \mam{}
    & \cellcolor{mycolor}\Y$^{T\ref{thm:mam}}$ 
    & \cellcolor{mycolor}\Y$^{T\ref{thm:mam}}$ & \N$^\texttt{a~}$
    & \cellcolor{mycolor} \Y$^{T\ref{thm:mam}}$ & \cellcolor{mycolor}\Y$^{T\ref{thm:mam}}$ & \N$^\texttt{a~}$  
    & \cellcolor{mycolor}\N$^{P\ref{prop:imp}}$ \\

    \mpbm{}
    & \cellcolor{mycolor}\Y$^{T\ref{thm:mpbm}}$ 
    & \cellcolor{mycolor}\Y$^{T\ref{thm:mpbm}}$ & \cellcolor{mycolor}\Y$^{T\ref{thm:mpbm}}$
    & \cellcolor{mycolor}\Y$^{T\ref{thm:mpbm}}$ & \cellcolor{mycolor}\N$^{R\ref{rmk:mpbm}}$ & \cellcolor{mycolor}\N$^{R\ref{rmk:mpbm}}$
    & \cellcolor{mycolor}\N$^{P\ref{prop:imp}}$ \\
    \hline
\end{tabular}
\end{center}
{\footnotesize Note: A `Y' indicates that the mechanism at that row satisfies the property at that column, and an `N' indicates that it does not. 
Results annotated with `\texttt{a}' follow from~\protect\cite{guo2021favoring}, `\texttt{b}' from~\protect\cite{hosseini2019multiple}, `\texttt{c}' from~\protect\cite{Bogomolnaia01:New}, `\texttt{d}' from~\protect\cite{aziz2020simultaneously}, and  `\texttt{e}' from~\protect\cite{kojima2009random} respectively. A result annotated with T, P, R refers to a Theorem, Proposition, or Remark in this paper, respectively.}
\end{table*}

The {\em Boston mechanism} (\bm{}) is widely used in the special case of the assignment problem where each agent is required to be matched with at most one item. The output of \bm{}  satisfies both \fcm{} and \opt{}~\cite{abdulkadirouglu2003school,Kojima2014:Boston}. Although \bm{} does not provide an ex-ante guarantee of envy-freeness, a variant of \bm{} named the {\em eager Boston mechanism} provided in our previous work~\cite{guo2021favoring} guarantees {\em \fsdwef{}} (\sdwef{}), a mildly weaker notion of ex-ante envy-freeness, while retaining \fcm{} and \opt{}. However, since each agent is matched with only one item, concerns over ex-post approximations of envy-freeness such as \efo{} do not arise.

For the general case of the assignment problem where agents may be assigned more than one item, recent work aims to achieve the ``{\em best of both worlds}'' (BoBW), i.e., both ex-ante and ex-post fairness.~\citeay{freeman2020best} showed that a form of ex-ante envy-freeness, \fsdef{} (\sdef), and ex-post \efo{} can be achieved simultaneously. \citeay{aziz2020simultaneously} showed that these fairness guarantees can also be achieved together with {\em \fsdopt{}} (\sdopt{}), an ex-ante variant of \opt{}.
However, the mechanisms developed in these works do not guarantee \fcm{}. 

The pursuits for efficiency and simultaneous ex-ante and ex-post fairness outlined above inevitably raise the following natural open question that we investigate in this paper:
{\em``How to design mechanisms that allocate $m$ items to $n$ agents and guarantee both efficiency (\fcm{} and \opt{}) and fairness (envy-freeness), both ex-ante and ex-post?''}

\paragraph{Our contributions.}
We provide two novel randomized mechanisms, the \fmam{} (\mam{}) and the \fmpbm{} (\mpbm), both of which satisfy ex-post \fcm{} and \opt{} together with different combinations of desirable efficiency and fairness properties  as we summarize in \cref{tab:properties}. In particular:
\begin{itemize}[wide,labelindent=0em]
    \item \mam{} provides both ex-post and ex-ante fairness guarantees, satisfying satisfies \sdwef{} and ex-post \efo{}  (\Cref{thm:mam}).
    \item \mpbm{} also satisfies ex-post \efo{}, and provides a stronger ex-ante efficiency guarantee (\sdopt{})  instead of ex-ante fairness (\Cref{thm:mpbm}).
\end{itemize}

\paragraph{Related work and Discussions.}
The assignment problem is a generalization of the matching problem with one-sided preferences~\cite{Moul04,manlove2013algorithmics}, where each agent must be assigned at most one item given agents' preferences over the items. 
In matching problems, \bm{}~\cite{abdulkadirouglu2003school} and \am{}~\cite{guo2021favoring} proceed in multiple rounds as follows. In each round $r$ of \bm{}, each agent that has not been allocated an item yet applies to receive its $r$-th ranked item. Each item, if it has applicants, is allocated to the one with the highest priority. Randomization over priority orders yields a mechanism whose expected output is a random assignment. In the round of \am{}, each remaining agent applies for its most preferred remaining item, and each item is allocated to one of the applicants through a lottery. This subtle change results in the ex-ante fairness guarantee of \am{}. 

The trade-off of efficiency and fairness in the assignment problem has been a long-standing topic, and guaranteeing \fcm{} with other properties has been the focus of several recent research efforts.
Notice that both \bm{} and \am{} guarantee \fcm{} since every item ranked on top by some agent is allocated to one such agent in the first round during the execution of both mechanisms.  
For the matching problems, \citeay{Ramezanian2021:Ex-post} showed that the efficiency notion \ffhr{} (\fhr{}), which characterizes \bm{} and implies \fcm{}.
Our previous work showed that \fhr{} is not compatible with \sdwef{}, and provided \ffhcr{} as an alternative  that also implies \fcm{} ~\cite{guo2021favoring}. 
For the general case of the assignment problem, \citeay{hosseini2021fair} showed  an assignment that satisfies both {\em\frkm{}} (a stronger efficiency property that also implies \fcm{}~\cite{irving2006rank}) and {\em\fefx{}} (a stronger approximation of envy-freeness that implies \efo{}~\cite{caragiannis2019unreasonable}) does not always exist.

Designing randomized mechanisms that provide ex-ante guarantees of efficiency and fairness has been a long standing concern in the literature. 
\citeay{Bogomolnaia01:New} proposed the probabilistic serial (PS) mechanism that satisfies \sdopt{} and \sdef{} for the matching problem. 
Similar efforts have been made for housing markets~\cite{athanassoglou2011house,altuntacs2022trading,yilmaz2010probabilistic} and assignment problems with quotas~\cite{Budish2013:Designing,kojima2009random}. 
The ex-post approximation of envy-freeness also raised concern in the previous work.
However, the endeavors to provide both ex-ante and ex-post guarantees of fairness simultaneously are more recent~\cite{babaioff2022best,freeman2020best,aziz2020simultaneously,hoefer2022best}.

\Cref{tab:properties} compares our \mam{} and \mpbm{} with existing mechanisms. The random serial dictatorship quota mechanism (RSDQ)~\cite{hosseini2019multiple} satisfies strategyproofness, meaning no agent can benefit by misreporting its preferences, but does not provide either ex-ante efficiency or ex-post fairness guarantees. 
\citeay{aziz2020simultaneously} proposed the PS-Lottery mechanism as an extension of PS which provides strong ex-ante efficiency and fairness guarantees. However, neither RSDQ nor PS-Lottery takes the ranks of items into consideration, and therefore they do not satisfy \fcm{}.
 
\section{Preliminaries}
An instance of the {\em assignment problem} is given by a tuple $(N,M)$, where $N=\{1,2,\dots,n\}$ is a set of $n$ {\em agents} and $M=\{o_1,o_2,\dots,o_m\}$ is a set of $m$ distinct indivisible {\em items}; and a {\em preference profile} $R=(\succ_j)_{j\in N}$, where for each agent $j\in N$, $\succ_j$ is a strict linear order representing $j$'s preferences over $M$. Let $\mr$ be the set of all possible preference profiles.

For each agent $j\in N$ and for any set of items $S\subseteq M$, we define $\rka{\succ_j}{o}{S}\in\{1,\dots,|S|\}$ to denote the rank of item $o\in S$ among the items in $S\subseteq M$ according to $\succ_j$, and $\tp{\succ_j}{S}\in S$ to denote the item ranked highest in $S$. When $S=M$, we use $\rk{\succ_j}{o}$ for short; and when agent $j$'s preference relation is clear from context, we use $\rka{j}{o}{S}$ and $\tp{j}{S}$ instead. We use $\succ_{-j}$ to denote the collection of preferences of agents in $N\setminus\{j\}$. For any linear order $\succ$ over $M$ and item $o$, $\ucs(\succ,o)=\{o'\in M|o'\succ o\}\cup\{o\}$ represents the items weakly preferred to $o$.
	
\vspace{1em}
\noindent{\bf Allocations, Assignments, and Mechanisms.} 
The subset of items that an agent receives, which we call an {\em \dal{}}, is a binary $m$-vector $a=[a(o)]_{o\in M}$.
The value $a(o)=1$ indicates that item $o$ is in the subset represented by $a$, and we also use $o\in a$ for that.
A {\em \ram{} allocation} is a $m$-vector $p=[p(o)]_{o\in M}$ with $0\le p(o)\le 1$, describing the probabilistic share of each item. Let $\Pi$ be the set of all possible \ram{} allocations, and any allocation belongs to $\Pi$ trivially.

An {\em \das{}} $A:N\to 2^M$ is a mapping from agents to \dals{}, represented by an $n\times m$ matrix. 
For each agent $j\in N$, we use $A(j)$ to denote the allocation for $j$. Let $\ma$ denote the set of all the \dass{}. A {\em \ram{} assignment} is an $n\times m$ matrix $P=[P(j,o)]_{j\in N, o\in M}$. 
For each agent $j\in N$, the $j$-th row of $P$, denoted $P(j)$, is agent $j$'s \ram{} allocation, and for each item $o\in M$, $P(j,o)$ is $j$'s probabilistic share of $o$. 
We use $\mP$ to denote the set of all possible \ram{} assignments, and we note that $\ma\subseteq\mP$, i.e., an \das{} can be regarded as a \ram{} assignment.

A {\em mechanism} $f:\mr\to\mP$ is a mapping from preference profiles to \ram{} assignments.
For any profile $R\in\mr$, we use $f(R)$ to refer to the \ram{} assignment output by $f$. 

\subsection{Desirable Properties}\label{sec:properties}

Before giving the specific definitions of desirable properties, we introduce two methods for random allocation comparison.

\begin{definition}\rm\cite{Bogomolnaia01:New}\label{dfn:sd}
	Given a preference relation $\succ$ over $M$, the {\em stochastic dominance} relation associated with $\succ$, denoted by $\sd{\null}$, is a partial ordering over $\Pi$ such that for any pair of \ram{} allocations $p,q\in\Pi$, $p$ (weakly) {\em stochastically dominates} $q$, denoted by $p\sd{\null} q$, if for any $o\in M$, $\sum_{o'\in\ucs(\succ,o)}p(o')\ge\sum_{o'\in\ucs(\succ,o)}q(o')$.
\end{definition}

\begin{definition}
	Given a preference relation $\succ$ over $M$, the {\em lexicographic dominance} relation associated with $\succ$, denoted by $\ld{\null}$, is a strict linear ordering over $\Pi$ such that for any pair of \ram{} allocations $p,q\in\Pi$, $p$  {\em lexicographically dominates} $q$, denoted by $p\ld{\null} q$, if there exists an item $o$ such that $p(o)>q(o)$ and $p(o')=q(o')$ for any $o'\succ o$.
\end{definition}

Given a preference profile $R$, an \das{} $A$ satisfies:
\begin{enumerate}[label=(\roman*),wide,labelindent=0em,topsep=0em,itemsep=0em]
	\item {\bf\fopt{} (\opt{})} if there does not exist another $A'$ such that $A'(j)\ld{} A(j)$ for $j\in N'\neq\emptyset$ and $A'(k)=A(k)$ for $k\in N\setminus  N'$,
    \item {\bf\fsdefo{} (\sdefo{})} if for any agents $j$ and $k$, there exists an item $o$ such that $A(j)\sd{j} A(k)\setminus{}\{o\}$, and
    \item {\bf\ffcm{} (\fcm)} if  there does not exist another $A'$ such that $\lvert \{j\in N\mid \rk{j}{o}=1\text{ and }o\in A'(j)\}\rvert>\lvert\{ j\in N\mid \rk{j}{o}=1\text{ and }o\in A(j)\}\rvert$.
\end{enumerate}
	
In general, given a property $X$ for \dass{}, a \ram{} assignment satisfies {\em ex-post} $X$ if it is a convex combination of \dass{} satisfying property $X$, and a mechanism $f$ satisfies a property $Y$ if $f(R)$ satisfies $Y$ for every profile $R\in\mr$. Given a preference profile $R$, a \ram{} assignment $P$ satisfies:
\begin{enumerate}[label=(\roman*),wide,labelindent=0em,topsep=0em,itemsep=0em]
    \item {\bf\fsdopt{} (\sdopt{})} if there is no \ram{} assignment $Q\neq P$ such that $Q(j)\sd{j}P(j)$ for any $j\in N$, and
    \item {\bf\fsdwef{} (\sdwef{})} if $P(k)\sd{j}P(j)\allowbreak\implies P(j)=P(k)$.
\end{enumerate}
	
A mechanism $f$ satisfies:
\begin{enumerate}[label=(\roman*),wide,labelindent=0em,topsep=0em,itemsep=0em]
    \item {\bf\fsdsp{} (\sdsp{})} if for every $R\in\mr$, any $j\in N$, and any $R'=(\succ'_j,\succ_{-j})$, it holds that $f(R')(j)\sd{j}f(R)(j)\implies f(R')(j)=f(R)(j)$,
    \item {\bf\ntr{}} if given any permutation $\pi$ over the items, $f(\pi(R))=\pi(f(R))$ for any preference profile $R$.
    The permutation $\pi$ is given as $\{(o_1,o_2),(o_2,o_3),\dots\}$, and $\pi(R)$ (respectively, $\pi(f(R))$) is obtained by replacing $o_i$ with $o_j$ in $R$ (respectively, $f(R)$) for each ordered pair $(o_i, o_j)\in \pi$.
\end{enumerate}

\begin{restatable}{lem}{lemopt}{}\label{lem:opt}{\em [Folklore]}
An \das{} $A$ satisfies \fopt{} if and only if the relation $\{(o,o')\in M\times M\mid\text{there\allowbreak{} exists } j\in N\text{ with }  o'\succ_j o\text{ and }o\in A(j)\}$ is acyclic.
\end{restatable}

\begin{restatable}{lem}{lemsdopt}{}\label{lem:sdopt}{\rm\cite{Bogomolnaia01:New}}
A random assignment $P$ satisfies \fsdopt{} if and only if the relation 
$\{(o,o')\in M\times M\mid\text{there\allowbreak{} exists an agent } j\in N \text{ where } o'\succ_j o\allowbreak{} \text{ and }P(j,o)>0\}$
is acyclic.
\end{restatable}

Due to \Cref{lem:opt,lem:sdopt}, \sdopt{} implies ex-post \opt{} for the assignment problem.

\begin{remark}
    [ex-post \fcm $\iff$ ex-ante \fcm{}]
    \rm
    A random assignment $P$ is ex-post \fcm{} if and only if it is ex-ante \fcm{}, i.e., the probabilistic shares of each item $o$ that is ranked first by some agent are allocated only to those agents who rank it first among all items. 
    This is because in any \das{} $A$ drawn from the probability distribution represented by $P$, each $o$ is assigned to one of the agents who rank it first if such an agent exists.
\end{remark}

\subsection{The Special Case of Matching}
The matching problem is a useful and important special case of the assignment problem where each agent must be matched with at most one item. 
To distinguish from the allocation in general cases, we call the \das{}  $A$ in the matching problem as a (one-to-one) matching where each agent $j$'s allocation $A(j)$ either consists of a single item or is the empty set.

 We introduce two efficiency notions that  imply \fcm{} and \opt{}.
{\em \Ffhr{}} (\fhr{}) requires that each item is allocated to an agent that ranks it as high as possible among all the items~\cite{Ramezanian2021:Ex-post}, while {\em \ffhcr{}} (\fhcr{}) requires that each remaining item is allocated to a remaining agent who prefers it to any other remaining items if such an agent exists~\cite{guo2021favoring}.
Formally, given a preference profile $R$, a matching $A$ satisfies
\begin{enumerate}[label=(\roman*),wide,labelindent=0em]
\item {\bf \ffhr{} (\fhr{})} if for any agents $j,k\in N$, $\rk{j}{A(j)}\le\rk{k}{A(j)}$ or $\rk{k}{A(k)}<\rk{k}{A(j)}$, and 
\item {\bf \ffhcr{} (\fhcr{})} if it holds that $o=\tp{A^{-1}(o)}{\allowbreak M\setminus\bigcup_{r'<r}\tps{A}{r'}}$ for every item $o\in \tps{A}{r}$ with integer $r\ge1$, where we define $\tps{A}{r}=\{o\in M\mid o=\tp{j}{\allowbreak M\setminus\bigcup_{r'<r}\tps{A}{r'}},$~\allowbreak$A(j)\notin\bigcup_{r'< r}\tps{A}{r'}\}$.
\end{enumerate}

\section{ Generalized Eager Boston Mechanism}

To achieve \fcm{} accompanied with the ex-ante and ex-post fairness in the assignment problem, we propose the \fmam{} mechanism (\mam{}) which is an extension of \am{}. 
We show in~\Cref{thm:mam} that \mam{} satisfies the efficiency requirements of \opt{} and \fcm{} ex-post while also providing fairness guarantees both ex-ante (\sdwef{}) and ex-post (\efo{}).

For any preference profile $R$ of $n$ agents' preferences over $m$ items, \mam{} (\Cref{alg:mam}) proceeds in $\lceil{m/n}\rceil$ rounds. 
At each round $c$, a matching is computed on the instance of the matching problem involving all $n$ agents and the remaining items using the \am{} mechanism. 
The final output $A$ of \mam{}$(N, M, R)$ is the composition of these matchings.
For convenience, at each round $c$ of \mam{}, we use 
\begin{itemize}[wide,labelindent=0em,label=-,itemsep=0pt,topsep=0pt]
    \item $M^c$ to refer to the items remaining at the beginning, and
    \item $A^c=\am{}(N, M^c, R)$ to refer to the matching output by \am{} for the instance $(N,M^c)$ of the matching problem.
\end{itemize}
We note that in the notation $\am{}(N, M^c, R)$, the preferences in $R$ are over $M$, so we have to specify the item set to be allocated $M^c$.
The final output of \mam{} is $A=\sum_{c=1}^{\lceil m/n\rceil}A^c$.

\begin{algorithm}[htb]
    \begin{algorithmic}[1]
        \State {\bf Input:} An assignment problem $(N,M)$ and a strict linear preference profile $R$.
        \State  $A\gets 0^{n\times m}$. $c\gets 1$.
        \For{$c=1$ to $\lceil{m/n}\rceil$}\hfill// {\texttt Round $c$ of \mam{}}
            \State $M^c\gets \{o\in M\mid A^{c'}({j,o})=0\text{ for all } j\in N, c'<c$\}.
            \Statex ~~~// Compute $A^c=\am{}(N,M^c)$
            \State $N'\gets N$, $A^c\gets 0^{n\times m}$.  
            
            \While{$N'\neq\emptyset$ and $M^c\neq\emptyset$}\hfill// {\texttt Round of \am{}}
                \For{each $o\in M^c$}
                    \State$N_o\gets\{j\in N'\mid\rk{j}{o}=\tp{j}{N'}\}$.
                    \State Pick $j_o$ from $N_o\neq\emptyset$ uniformly at random.
                    \State $A^c({j_o,o})\gets1$.
                    \State $M^c\gets M^c\setminus\{o\in M^c\mid N_o\neq\emptyset\}$.
                    \State $N'\gets N'\setminus{\cup_{o\in M^c}\{j_o\}}$.
                \EndFor
            \EndWhile
            \State $A\gets A+A^c$. 
        \EndFor
        \State \Return $A$
    \end{algorithmic}
    \caption{\label{alg:mam}Generalized Eager Boston Mechanism (\mam{})}
\end{algorithm}

    \begin{example}\label{eg:mam}\rm
        We illustrate the execution of \mam{} (\Cref{alg:mam}) using an instance with four items and two agents with preferences:
        \begin{equation*}
            \begin{split}
            \succ_{1}\text{: }& a \succ_1 b \succ_1 c \succ_1 d,\\
            \succ_{2}\text{: }& a \succ_2 c \succ_2 b \succ_2 d.\\
            \end{split}
        \end{equation*}

        The table below illustrates the execution of \Cref{alg:mam}.
        Each entry in the table refers to the item the row agent applied for in a round of execution of \am{} within a round of \mam{} indicated by the column. Circled entries indicate allocations, and a `/' entry indicates the agent was allocated an item in a previous round of \am{} and does not apply for any item.

        \begin{center}
    		\centering
    		\begin{tabular}{c|cc|cc}
    			Round of \mam{} &   \multicolumn{2}{c|}{$1$} & \multicolumn{2}{c}{$2$}\\\hline
                    Round of \am{} & $1$ & $2$ & $1$ & $2$\\
    			\hline
    			Agent $1$ & \bcircled{$a$} & / & \bcircled{$b$} & / \\
    			Agent $2$ & $a$ & \bcircled{$c$} & $b$ & \bcircled{$d$}\\
    			\hline
    		\end{tabular}
    	\end{center}
        
        \noindent{\bf Round $1$ of $\mam{}(N,M,R)$:} the \am{} mechanism is executed on the matching instance $(N=\{1,2\},M^1=\{a,b,c,d\})$ involving all items and both agents and outputs the assignment $A^1:1\gets(a), 2\gets(c)$ as follows.
        \begin{itemize}[wide,label=-,labelindent=0em,topsep=0em,itemsep=0em]
        \item At round $1$ of $\am{}(N,M^1,R)$, both agents $1$ and $2$ apply for their top-ranked item $a$ among items in $M^1$ and are in the applicant set $N_a$. Suppose that agent $1$ is chosen to receive $a$ and therefore removed from the current round of \mam{}. 
        \item At round $2$ of $\am{}(N,M^1,R)$, agent $2$ applies for her top item $c$ among the remaining items $\{b,c,d\}$ and gets it.
        \end{itemize}

        \noindent{\bf Round $2$ of $\mam{}(N,M,R)$:} \am{} is executed on the matching instance $(N=\{1,2\},M^2=\{b,d\})$ and outputs $A^2:1\gets(b), 2\gets(d)$. 
        \begin{itemize}[wide,label=-,labelindent=0em,topsep=0em,itemsep=0em]
        \item At round $1$ of $\am{}(N,M^2,R)$, agents $1$ and $2$ both prefer $b$ over $d$ and apply for $b$. Now suppose agent $1$ gets $b$.
        \item At round $2$ of $\am{}(N,M^2,R)$, agent $2$ applies for and is allocated the only remaining item $d$. 
        \end{itemize}
        
        Together the execution above outputs the \das{} $A:~1\gets(a,b),2\gets(c,d)$.
     \hfill $\square$
    \end{example}

    Notice that each of the matchings $A^c=\am{}(N,M^c, R)$ satisfy \fhcr{} according to Theorem~1 in our previous work~\cite{guo2021favoring}. The first-choice maximality of $A$ follows due to the matching $A^1$.

    \begin{restatable}{cor}{clmthmmam1}\label{clm:thm:mam:1}
         Given any preference profile $R$ and assignment $\mam{}(R)=\sum_{c=1}^{\lceil{m/n}\rceil}A^c$, for each $c\in\{1,\dots,\lceil{m/n}\rceil\}$, $A^{c}$  satisfies \fhcr{} for the matching problem $(N,M^c)$.
    \end{restatable}

    \Cref{clm:mfhcr} shows that given a round  of \mam{} , every agent prefers the item they are allocated in that round to all the items allocated in any subsequent round. \efo{} follows from this construction.

    \begin{restatable}{lem}{clmmfhcr}{}\label{clm:mfhcr}   
    Given any preference profile $R$ and assignment $\mam{}(R)=\sum_{c=1}^{\lceil{m/n}\rceil}A^c$, for any pair of agents $j,k$ and integer $c\in\{1,\dots,\lceil{m/n}\rceil-1\}$, $A^c(j)\succ_j A^{c+1}(k)$.
    \end{restatable}

    With \cref{clm:thm:mam:1,clm:mfhcr}, we present the following theorem to show the properties of \mam{}.

    \begin{restatable}{thm}{thm:mam}{}\label{thm:mam}
		The generalized eager Boston mechanism satisfies ex-post \opt{}, ex-post \sdefo{},  ex-post \fcm{}, and \sdwef{}.
	\end{restatable}
    \begin{proof}
        Given any preference profile $R$, let $A=\mam{}(R)$.

        \paragraph{\opt{}:}
        We prove by mathematical induction that for any $c\in\{1,\dots,\lceil m/n\rceil\}$, assignment $\sum_{c'\le c}A^{c'}$ satisfies \opt{}.
       
        \noindent{\em Base case: }For $c=1$, $A^1$ satisfies \fhcr{} by~\Cref{clm:thm:mam:1}.
        Since \fhcr{} implies \opt{} (Proposition 2 of \citeay{guo2021favoring}), we have that $A^1$ trivially satisfies \opt{}.

        \noindent{\em Inductive step: }
        For any $c>1$, given that $A' = \sum_{c'<c}A^{c'}$ satisfies \opt{}, we prove that $A = \sum_{c\le c}A^{c}$ also satisfies \opt{}.
        Suppose for the sake of contradiction that there exists a cycle in $A$ by~\Cref{lem:opt}.
        By the given condition and~\Cref{clm:thm:mam:1}, we know that $A'$ and $A^{c}$  satisfies \opt{} and therefore there is no cycle in any one of them, which means that the cycle in $A$ must involve both $A'$ and $A^{c}$.
        Therefore there must exist a pair of items $o=A^{c'}(j)$ and $o'=A(k)^{c}$ such that $o'\succ_j o$ for some agents $j,k$ and $c'<{c}$, which contradicts~\Cref{clm:mfhcr} which implies that $A^{c'}(j)\succ_j A^{c}(k)$.

        \paragraph{\sdefo{}:}
        By~\Cref{clm:mfhcr}, $A^c(j)\succ_j A^{c+1}(k)$ for any agents $j$ and $k$.
        In this way, $A(j)\sd{j}A(k)\setminus\{A^1(k)\}$.

        \paragraph{\fcm{}:}
        By \Cref{clm:thm:mam:1}, we know $A^1$ satisfies \fhcr{} for the assignment problem $(N,M,R)$.
        For any item $o\in M$, $o$ is ranked top by some agent if and only if $o\in\tps{A^1}{1}$. 
        Then for any such item $o$ and for agent $j$ with $A(j)=o$, we have $o=\tp{j}{M}$ by the definition of \fhcr{}, which implies $A$ satisfies \fcm{}.

        \paragraph{\sdwef{}:} Let $P=\mathbb{E}($\mam{}$(R))$. Recall that $A^c$ is defined to be the outcome of $\am{}(N,M^c)$ in round $c$ in the execution of \mam{}.
        Let $\ma^{<c}$ be the set of all the possible intermediate outcomes of first  $c-1$ rounds of \mam{}. For each $A^{<c}\in\ma^{<c}$, let $\alpha(A^{<c})$ be the probability that $A^{<c}$ is output after the first $c-1$ rounds. 
        We define $\alpha(A^c\mid A^{<c})$ to be the probability that $\am{}(N,M^c)=A^c$ is the matching output by \am{} in round $c$ of \mam{} given $A^{<c}$ as the intermediate outcome of the first $c-1$ rounds of \mam{}. 
        Let $\ma^c_{|A^{<c}}$ be the set of all possible matchings given the instance $(N,M^c)$. Let $P^c$ be the expected assignment computed at the end of round $c$ of \mam{}. Then, we have $P=\sum_{c=1}^{\lceil m/n\rceil}P^c$ and:
        \begin{equation}\label{eq:thm:mam:2}
            P^c=\sum\limits_{A^{<c}{}\in\ma^{<c}}\sum\limits_{A^c\in\ma^c_{|A^{<c}}}^{}\alpha(A^c)*\alpha(A^c|A^{<c})*A^c.
        \end{equation}
        
        We show that for each $c\in\{1,\dots,\lceil m/n\rceil\}$, it holds for any pair of agents $j,k\in N$ that if $P^c(j)\neq P^c(k)$ then $P^c(j)\ld{j} P^c(k)$.

        By Eq~(\ref{eq:thm:mam:2}), in an arbitrary round $c$ of \mam{}, for any assignment $A^{<c}$ computed in the first $c-1$ rounds of \mam{}, let $Q^c_{|A^{<c}}=\sum_{A^c\in\ma^c_{|A^{<c}}}\alpha(A^c\mid A^{<c})*A^c$ denote the expected output of $\am{}(N,M^c)$ given $A^{<c}$, where $M^c=\{o\in M\mid\sum_{j\in N}A^{<c}(j,o)=0\}$.

        We first show that:
        \begin{equation}\label{eq:thm:mam:1}
            Q^{c}_{|A^{<c}}(j)\ld{j} Q^{c}_{|A^{<c}}(k)\text{ if }Q^{c}_{|A^{<c}}(j)\neq Q^{c}_{|A^{<c}}(k).
        \end{equation}
        At any round $r$ of \am{} within round $c$ of \mam{}, let $o$ be the item that agent $j$ applies for.
        There are two cases according to the item agent $k$ applies for:
        \begin{enumerate*}[label=(Case \arabic*)]
            \item Suppose agent $k$ applies for the same item as $j$. Then they have the same probability to get item $o$, i.e. $Q^{c}_{|A^{<c}}(j,o)=Q^{c}_{|A^{<c}}(k,o)$ according to line~9 of \Cref{alg:mam}.
            \item Suppose agent $k$ applies for a different item $o'$. Then $o\succ_j o'$ and agent $k$ has no chance of receiving item $o$, since item $o$ must be allocated to one of the agents that applied for $o$ in round $r$ of \am{}, and it follows that $Q^{c}_{|A^{<c}}(j,o)>Q^{c}_{|A^{<c}}(k,o)=0$.
        \end{enumerate*}
        Together, both cases leads to Eq~(\ref{eq:thm:mam:1}).

        By Eq~(\ref{eq:thm:mam:2}), we have that $P^c=\sum_{A^{<c}\in\ma^{<c}}\alpha(A^{<c})*Q^c_{|A^{<c}}$. Together with \Cref{clm:mam} below, we have by Eq~(\ref{eq:thm:mam:1}) that $P^{c}(j)\ld{j} P^{c}(k)$ if $P^{c}(j)\neq P^{c}(k)$. 
        
        \begin{restatable}{claim}{clmgam}{}\label{clm:mam}
            Given random allocations $p_1,\dots,p_s$ and $q_1,\dots,q_s$ with $p_{i}\ld{} q_{i}$ or $p_i=q_i$ for any integer $i\in [1,s]$.
            If $\sum_{i=1}^sp_{i}\neq\sum_{i=1}^sq_{i}$, then $\sum_{i=1}^sp_{i}\ld{}\sum_{i=1}^sq_{i}$ .
        \end{restatable}

        By $P=\sum_{c=1}^{\lceil m/n\rceil}P^c$ and~\Cref{clm:mam}, we have that $P(j)\ld{j}P(k)$ when $P(j)\neq P(k)$.
        It follows that if $P(k)\sd{j}P(j)$, then $P(j)\ld{j}P(k)$ does not hold and therefore $P(j)=P(k)$, which means that  $P$ satisfies \sdwef{}.
        \end{proof}

        \begin{remark}\label{rmk:mam}\rm
        \mam{} does not satisfy \sdopt{} and \sdef{}.
        For the matching problem, \mam{} executes lines~3-13 once and therefore it is equivalent to \am{}.
        According to Proposition~15 of~\citeay{guo2021favoring}, \am{} does not satisfy \sdopt{} and \sdef{}, and therefore neither does \mam{}.
        \end{remark}
	
\section{Generalized Probabilistic Boston Mechanism}
     In this section, we propose \fmpbm{} (\mpbm{}) mechanism. 
     \mpbm{} also satisfies  \fcm{} and  \efo{} ex-post, and additionally provides an ex-ante efficiency guarantee of \sdopt{} as we show in \Cref{thm:mpbm} later, but it does not provide an ex-ante fairness guarantee (\Cref{rmk:mpbm}). In comparison, \bm{} cannot satisfy \efo{} since \fhr{} is not compatible with \efo{} (See~\Cref{prop:imp_fhr} in Appendix, and \mam{} satisfies \sdwef{} but not \sdopt{}.
    
\begin{algorithm}[htb]
    \begin{algorithmic}[1]
        \State {\bf Input:} An assignment problem $(N,M)$ and a strict linear preference profile $R$.
        \State For each $o\in M$, $s(o)\gets 1$. $M'\gets M$. $c\gets0$.
        \While{$M'\neq\emptyset$}\hfill// {\texttt Round $c$ of \mpbm{}}
            \State $c\gets c+1$. $P^c\gets 0^{n\times m}$. $N'\gets N$. $r\gets 1$.
            \While{$M'\neq\emptyset$ and $N'\neq\emptyset$}\hfill// {\texttt Consumption round $r$}
                \For{each $o\in M'$}
                    \State $N_o\gets\{j\in N'\mid\rk{j}{o}=r\}$.
                    \State \multiline{
                        Each agent $j\in N_o$ consumes $o$ at an equal rate. \\
                        Let $\delta_j$ be the amount of item $o$ consumed by $j$.
                        \begin{enumerate}[label={\footnotesize 8.\arabic*:},topsep=0pt]
                            \item  Agent $j$ stops consumption when either
                            \begin{itemize}[label=-,topsep=0pt]
                                \item $\sum_{o'\in\ucs(j,o)}P^c(j,o')+\delta_j=1$, or
                                \item $\bigcup_{k\in N_o}\delta_k=s(o)$.
                            \end{itemize} 
                            \item $P^c(j,o)\gets \delta_j$, and $s(o) \gets s(o) - \bigcup_{k\in N_o}\delta_k$.
                        \end{enumerate}
                    }
                    \vspace{-1em}
                    \State \multiline{
                        $M'\gets M'\setminus\{o\in M'\mid s(o)=0\}$.
                        $N'\gets N'\setminus\{j\in N'\mid \sum_{o'\in\ucs(j,o)}P^c(j,o')=1\}$.
                    }
                \EndFor
                \State $r\gets r+1$.
            \EndWhile
        \EndWhile
        
        \State \Return $\sum\limits_{c=1}^{\lceil{m/n}\rceil}P^c$.
    \end{algorithmic}
    \caption{\label{alg:mpbm} Generalized Probabilistic Boston Mechanism (\mpbm)}
\end{algorithm}

\mpbm{}, defined in \Cref{alg:mpbm}, proceeds by assigning the $m$ distinct indivisible items in $\lceil{m/n}\rceil$ rounds. 
Each item $o$ initially has $s(o)=1$ unit of supply to be consumed.
In each round $c$ (lines~3-10 of~\Cref{alg:mpbm}), each agent is allowed to consume, i.e., be allocated, at most one unit of items cumulatively over multiple consumption rounds as follows. In each consumption round $r$ (lines~5-10), for each item $o$, all of the agents who rank item $o$ in position $r$ over all items, represented by the set $N_o$ in line~7, consume item $o$ at an equal rate. An agent $j\in N_o$  quits consuming $o$ when either the supply of item $o$ is exhausted, or $j$ has cumulatively consumed one unit of items in round $c$ of \mpbm{}. \mpbm{} terminates when all the items are consumed to exhaustion of their supply and returns the probability shares of items that agents consume during execution. We use $P^c$ to refer to the random assignment computed at the end of each round $c\in\{1,\dots,\lceil m/n\rceil\}$ of \mpbm{}.

\begin{example}\label{eg:mpbm}\rm
    We use the instance in~\Cref{eg:mam} to illustrate the execution of \mpbm{} in~\Cref{fig:mpbm}.

    \begin{figure}[h]
        \centering
        \includegraphics[width=\linewidth]{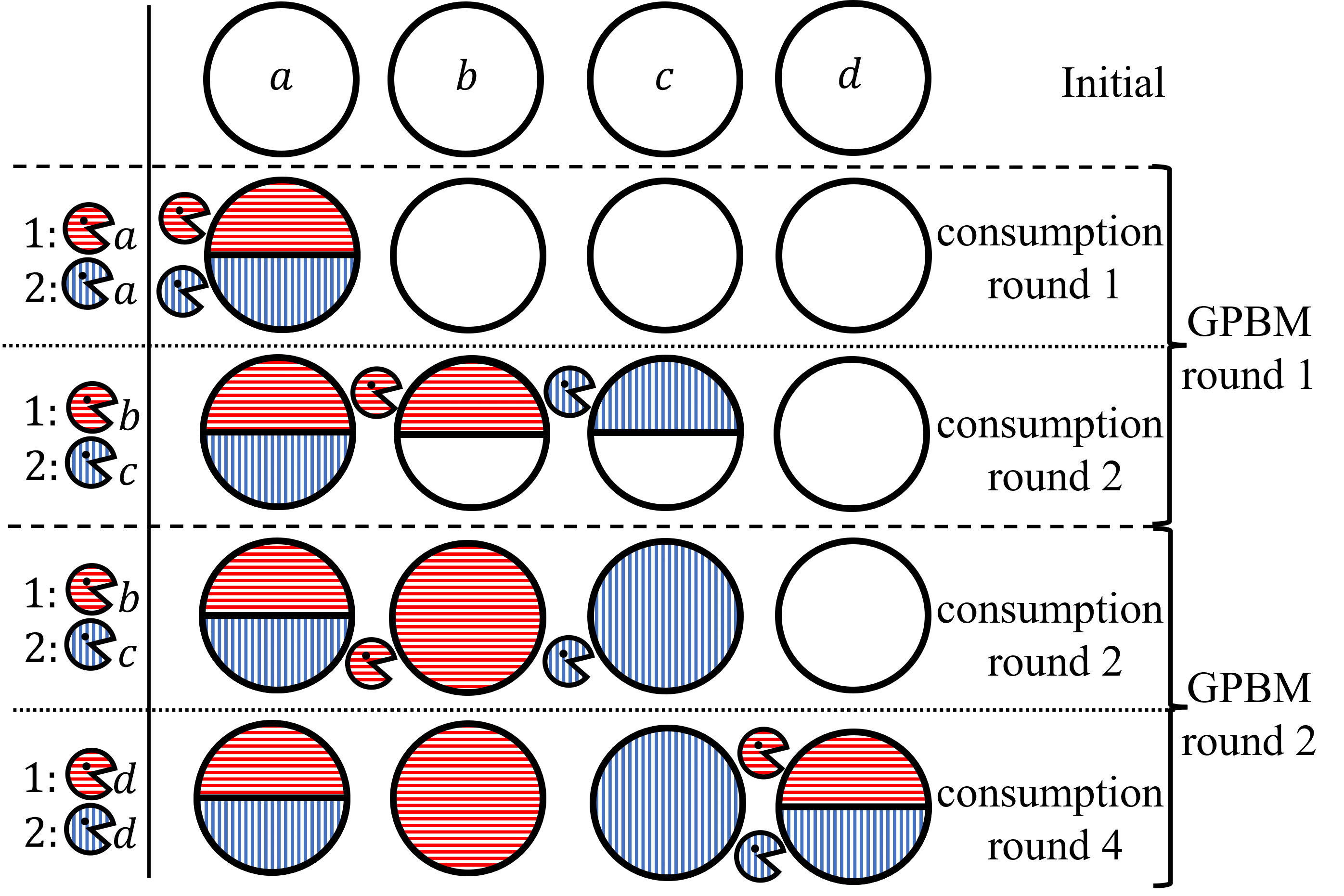}
        \caption{An example of the execution of \mpbm{}.}
        \label{fig:mpbm}
    \end{figure}

    \noindent{\bf Round $1$ of \mpbm{}:}  $P^1$ is generated as follows.
    \begin{itemize}[wide,label=-,labelindent=0em,itemsep=0pt,topsep=0pt]
        \item At consumption round $1$, agents $1$ and $2$ both consume item $a$ at an equal rate, and therefore $P^1(1,a)=P^1(2,a)=1/2$.
        \item At consumption round $2$, agent $1$ consumes item $b$ while agent $2$ consumes $c$, which results in $P^1(1,b)=P^1(2,c)=1/2$. Notice that $b$ and $c$ are not exhausted, but consumption stops since agents $1$ and $2$ have each cumulatively consumed one unit.
    \end{itemize}

    \noindent{\bf Round $2$ of \mpbm{}:} 
    We omit the consumption rounds that do not allocate any items in generating $P^2$.
    \begin{itemize}[wide,label=-,labelindent=0em,itemsep=0pt,topsep=0pt]
        \item At consumption round $2$, agent $1$ consumes item $b$ while agent $2$ consumes $c$, which results in $P^2(1,b)=P^2(2,c)=1/2$.
        \item At consumption round $4$, agents $1$ and $2$ both consume item $d$ and split it equally, i.e., $P^2(1,d)=P^2(2,d)=1/2$.
    \end{itemize}
 Then we obtain $P^1$ and $P^2$ in the following:
 
\noindent
\begin{minipage}{\linewidth}
    \begin{minipage}{0.4\linewidth}
    \begin{tabular}{c|cccc}
        \multicolumn{5}{c}{Assignment $P^1$}\\
        &  $a$ & $b$ & $c$ & $d$\\\hline
        1 & $1/2$ & $1/2$ & $0$ & $0$\\
        2 & $1/2$ & $0$ & $1/2$ & $0$ \\
    \end{tabular}
    \end{minipage}
    \hspace{0.1\linewidth}
    \begin{minipage}{0.4\linewidth}
    \centering
    \begin{tabular}{c|cccc}
        \multicolumn{5}{c}{Assignment $P^2$}\\
        &  $a$ & $b$ & $c$ & $d$\\\hline
        1 & $0$ & $1/2$ & $0$ & $1/2$\\
        2 & $0$ & $0$ & $1/2$ & $1/2$ \\
    \end{tabular}
    \end{minipage}
\end{minipage}
\end{example}

    We show in \Cref{clm:gpbm} that the random assignment $P^c$ obtained at the end of each round $c$ of \mpbm{} satisfies \sdopt{}, which will be instrumental in proving  that the output of \mpbm{} always satisfies \sdopt{} in \Cref{thm:mpbm}.

    \begin{restatable}{lem}{clmgpbm}{}\label{clm:gpbm}
        Given any preference profile $R$, for every $c\in\{1,\dots,\lceil m/n\rceil\}$, the assignment $P^c$ computed at the end of round $c$ of \mpbm{} satisfies \sdopt{}.
       
    \end{restatable}

An \das{} can be generated from the output $P$ of \mpbm{} using the algorithm suggested by \citeay{aziz2020simultaneously}, which we define in \Cref{alg:generate} in Appendix.
\Cref{alg:generate} proceeds\allowbreak{} by first creating for each agent $j$ the subagents $j^1,\dots,j^{\lceil{m/n}\rceil}$ who have the random allocations $P^1(j),\dots,P^{\lceil{m/n}\rceil}(j)$ respectively. The $n*\lceil m/n\rceil$ subagents thus created then participate in a lottery, and each agent is allocated the items won by all of its subagents, e.g., the \das{} $A$ in~\Cref{eg:mam} is a possible result drawn from $\{P^1,P^2\}$.
Since $\sum_{o\in M}P^c(j,o)\le 1$, each subagent does not obtain more than one item.
For any $A$ drawn from $P$, we define for $c=1,\dots,{\lceil{m/n}\rceil}$ that 
\begin{itemize}[wide,labelindent=0em,itemsep=0pt,topsep=0pt,label=-]
    \item $A^{c}$ to refer to the one-to-one matching  over all the subagents $j^c$ in $A$, and
    \item $M^c=\{o\in M\mid o\neq A^{c'}(j)\text{ for some }c'<c \text{ and }j\in N\}$ to be the set of items that are not allocated in the matching $A^{c'}$ with any $c'<c$.
\end{itemize}
In the following \Cref{clm:thm:mpbm:1}, we prove that each $A^c$ satisfies \fhr{}  and  that each agent prefers the items with positive shares in her allocation in the current round to any item to be allocated in later rounds.

        \begin{restatable}{lem}{clmthmmpbm}\label{clm:thm:mpbm:1}
             Given any preference profile $R$ and $\{A^1,\dots,A^{\lceil{m/n}\rceil}\}$ drawn from $\mpbm{}(R)=\sum_{c=1}^{\lceil{m/n}\rceil}P^c$ by~\Cref{alg:generate}, 
            \begin{enumerate}[label=\rm(\roman*),wide,labelindent=0em,itemsep=0pt,topsep=0pt]
                \item $A^{c}$ satisfies \fhr{} for the assignment problem $(N,M^c)$ for $c=1,\dots,\lceil{m/n}\rceil$;
                \item $A^c(j)\succ_j A^{c+1}(k)$ for $c=1,\dots,\lceil{m/n}\rceil-1$.
            \end{enumerate}
        \end{restatable}

    \begin{restatable}{thm}{thm:mpbm}{}\label{thm:mpbm}
		The generalized probabilistic Boston mechanism satisfies ex-post \fcm{}, ex-post \sdefo{}, and \sdopt{}.
	\end{restatable}

    \begin{proof}

        Let $A$ be an \das{} drawn from the distribution represented by $\{P^1,\dots, P^{\lceil{m/n}\rceil}\}$.
        
        \paragraph{\fcm{}:}
        Assume that $A$ does not satisfy \fcm{}. Then there exists an item $o$ such that a set of agents $N'\subseteq N$ prefer $o$ most, but $o$ is allocated to an agent $k\not\in N'$, i.e., $o\in A(k)$. Consider any agent $j\in N'$ and let $o'=A^1(j)$ and $o=A^c(k)$ for some $c$. 
        Since $o\succ_j o'$, it must hold that $c=\rka{k}{o}{A(k)}=1$. Otherwise, if $c>1$, then it means that $o=A^c(k)\succ_j A^1(k)$, a contradiction to \Cref{clm:thm:mpbm:1}~(\romannumeral2).
        However, we have that $\rk{k}{o}>1=\rk{j}{o}$ by the selection of $j$ and $k$, which contradicts the fact that $A^1$ satisfies \fhr{} by~\Cref{clm:thm:mpbm:1}~(\romannumeral1).

        \paragraph{\sdefo{}:}
        By~\Cref{clm:thm:mpbm:1}~(\romannumeral2), $A^c(j)\succ_j A^{c+1}(k)$ for any agents $j$ and $k$. It follows that $A(j)\sd{j} A(k)\setminus\{A^1(k)\}$.
    
        \paragraph{\sdopt{}:}
        Let $P=\mathbb E(\text{\mpbm{}}(R))=\sum_{c\le\lceil{m/n}\rceil}P^c$.
        We prove by mathematical induction that for any $c\in\{1,\dots,\lceil{m/n}\rceil\}$, it holds that $P^{\le c}=\sum_{c'\le c}P^{c'}$ satisfies \sdopt{}.
        
        \noindent{\em Base case: }
        For $c=1$, $P^{\le c}=P^1$ and therefore it satisfies \sdopt{} trivially by~\Cref{clm:gpbm}.

        \noindent{\em Inductive step: }
        Suppose that $P^{\le c}=\sum_{c'\le c}P^{c'}$ satisfies \sdopt{}. 
        Now, assume for the sake of contradiction that for $P^{\le c+1}$, there exists a cycle in the relation in~\Cref{lem:sdopt}. Notice that by \Cref{clm:gpbm}, $P^{c+1}$ satisfies \sdopt{}. Together with the assumption that $P^{\le c}$ satisfies \sdopt{}, this means the cycle must involve items with positive shares in both $P^{\le c}$ and $P^{c+1}$, i.e., there exist items $o,o'$ and a pair of agents $j,k$ involved in the cycle such that:
        \begin{equation}\label{eq:thm:mpbm:1}
            o\succ_k o'\text{, }P'(k,o')>0\text{, and }P^{c+1}({j,o})>0.
        \end{equation}

        By Eq~(\ref{eq:thm:mpbm:1}), it must hold that in round $c+1$ of \mpbm{} when $P^{c+1}$ is generated, agent $j$ consumes the item $o$ that agent $k$ prefers to the item $o'$.
        We also note that agent $k$ consumes $o'$ in a strictly earlier round $c' \le c$ of \mpbm{}.
        By line~9 of \Cref{alg:mpbm}, this implies that $s(o)>0$ at the beginning of round $c'$ of \mpbm{}. Then, by lines~7 and 8, for any item $o''$ consumed by agent $k$, i.e., where $P^{c'}({k,o''})>0$, we have that either $o''\succ_k o$ or $o'' = o$, a contradiction to Eq~(\ref{eq:thm:mpbm:1}). 
        
        Thus by induction, $P=\sum_{c\le\lceil{m/n}\rceil}P^c$ satisfies \sdopt{}.
        \end{proof}

        \begin{remark}\label{rmk:mpbm}\rm
            \mpbm{} does not satisfy \sdwef{}.
            For the assignment problem with following profile $R$, \mpbm{} outputs assignment $P$.
            
\noindent
\begin{minipage}{\linewidth}
    \begin{minipage}{0.3\linewidth}
        \begin{equation*}
                \begin{split}
                &\succ_{1,2}\text{: } a\succ b\succ c\succ d,\\
                &\succ_3\text{: } a\succ_3 d\succ_3 b\succ_3 c,\\
                &\succ_4\text{: } d\succ_4 a\succ_4 b\succ_4 c.\\
                \end{split}
            \end{equation*}
    \end{minipage}
    \hspace{0.12\linewidth}
    \begin{minipage}{0.5\linewidth}
            \begin{tabular}{r|cccc}
                \multicolumn{5}{c}{Assignment $P$}\\
                &  $a$ & $b$ & $c$ & $d$\\\hline
                1,2 & $1/3$ & $1/2$ & $1/6$ & $0$\\
                3 & $1/3$ & $0$ & $2/3$ & $0$\\
                4 & $0$ & $0$ & $0$ & $1$\\
            \end{tabular}
    \end{minipage}
\end{minipage}

\vspace{0.5em}
We see that $\sum_{o'\succ_3 o}P(1,o')=\sum_{o'\succ_3 o}P(3,o')$ for $o\in\{a,c,d\}$ and $\sum_{o'\succ_3 b}P(1,o')=5/6>1/3=\sum_{o'\succ_3 b}P(3,o')$.
It follows that $P(1)\neq P(3)$ and $P(1)\sd{3}P(3)$, which violates \sdwef{}.
        \end{remark}

    \section{An Impossibility Result}

    In \Cref{prop:imp}, we show that for the mechanisms that satisfy \fcm{}, \opt{}, and \efo{} ex-post, they are impossible to guarantee strategyproofness (\sdsp{})  without violating neutrality, which requires that any permutation of the item labels results in an assignment where the items allocated to each agent are permuted in the same manner. Therefore, \mam{} and \mpbm{} trivially satisfy neutrality, and therefore they cannot provide guarantee of strategyproofness.

    \begin{restatable}{prop}{propimp}{}\label{prop:imp}
        There is no \sdsp{} mechanism which simultaneously satisfies  \fcm{}, \opt{}, \sdefo{}, and \ntr{} ex-post.
    \end{restatable}

    \begin{proof}
        Suppose that $f$ is an \sdsp{} mechanism that satisfies  \fcm{}, \opt{}, \sdefo{}, and \ntr{} ex-post.
        Let $R$ be:
        \begin{equation*}
            \begin{split}
            \succ_1\text{: }& a\succ_1 b\succ_1 c\succ_1 d,\\
            \succ_2\text{: }& d\succ_2 a\succ_2 b\succ_2 c.
            \end{split}
        \end{equation*}
        
        In any assignment satisfying \opt{}, agent $1$ cannot obtain item $d$.
        We note \sdefo{} requires that an agent can obtain one more item  than any other at most.
        Then with this condition, the \opt{} assignments for $R$ are:
        \begin{equation*}
            \begin{split}
                A_1:&~1\gets \{a,b\},2\gets \{c,d\},\\
                A_2:&~1\gets \{a,c\},2\gets\{b,d\},\\
                A_3:&~1\gets\{b,c\},2\gets\{a,d\}.
            \end{split}
        \end{equation*}
        We observe that $A_1$ and $A_2$ also satisfy \fcm{} and \sdefo{}, and $A_3$ satisfies \sdefo{} here but violates \fcm{}.
        Then we have that $f(R)=\alpha_1*A_1+\alpha_2*A_2$, denoted $P$.

        Let $R'=(\succ_1,\succ'_2)$ be the profile obtained from $R$ when agent $2$ misreports her preferences as $\succ'_2$ below:
        \begin{equation*}
            \begin{split}
            \succ'_2\text{: }& a\succ'_2 b\succ'_2 d\succ'_2 c\\
            \end{split}
        \end{equation*}
        In any assignment satisfying \sdefo{} for $R'$, each agent must get only one item in $\{a,b\}$. Moreover, \opt{} requires that $1$ gets $c$ and $2$ gets $d$.
        In this way, only $A_2$ and $A_3$ satisfy \fcm{}, \opt{}, and \sdefo{} for the profile $R'$.
        Then we have that $f(R')=\alpha'_2*A_2+\alpha'_3*A_3$, denoted $P'$.
        We also observe that if $\alpha_1\neq0$ or $\alpha'_3\neq0$, then $P'\neq P$ and $P'_2\sd{2} P_2$ since $A_3\sd{2} A_2\sd{2} A_1$.
        Since $f$ satisfies \sdsp{}, we must have that $P'_2 = P_2 = A_2(2)$ which means that $\alpha_1=\alpha'_3=0$, and therefore $P'=P=A_2$.

        Let $\pi=\{(c,d),(d,c)\}$ be a permutation on $M$ that swaps the labels of items $c$ and $d$.
        We observe that $\pi(R')$ can be obtained by swapping the preferences of agents $1$ and $2$ in $R'$. Therefore $f(\pi(R'))$ can be constructed in the same manner as the assignment above, and it can be obtained by swapping the allocations of agents $1$ and $2$ in $P'=A_2$.
        We also have that $\pi(f(R'))=\pi(P')=\pi(A_2)$.
        Both assignments are shown below:
        \begin{equation*}
            \begin{split}
                f(\pi(R')):&~1\gets\{b,d\},2\gets\{a,c\},\\
                \pi(f(R')):&~1\gets \{a,d\},2\gets\{b,c\}.
            \end{split}
        \end{equation*}
        
        It is easy to see that $f(\pi(R'))\neq\pi(f(R'))$, meaning that $f$ must violate neutrality, a contradiction.
    \end{proof}

\section{Conclusion and Future Work}
Our results contribute towards the efforts to achieve both ex-ante and ex-post guarantees of both efficiency and fairness in the assignment of indivisible items.
We have provided the first mechanisms that satisfy ex-post \fcm{}, \opt{}, and \efo{} simultaneously.
In terms of the ex-ante guarantee, our \mam{} is fair, while \mpbm{} is efficient.
We have also  shown that  mechanisms of this kind cannot be strategyproof.

We wonder whether it is possible to achieve stronger  efficiency or fairness properties.
Recent works on identifying domain restrictions, under which certain impossibility results no longer pose a barrier and allow for mechanisms with stronger guarantees~\cite{hosseini2021fair,wang2023multi}, is a promising avenue for such investigations. In general, finding what combinations of  properties can be satisfied simultaneously, and what constitutes the BoBW, is an ongoing pursuit. 
Some works on designing mechanisms with desirable properties under constraints~\cite{garg2010assigning,budish2017course} that reflect real-world considerations such as agents' quotas~\cite{aziz2022vigilant,balbuzanov2022constrained} and more generally, those involving matroid constraints~\cite{dror2021fair,biswas2018fair,biswas2019matroid} are another interesting direction for future research.


\clearpage

\section*{Acknowledgments}
XG acknowledges the National Key R\&D Program of China under Grant 2021YFF1201102 for support. LX acknowledges NSF \#1453542, \#2007994, \#2106983, and a Google Research Award for support. YC acknowledges NSFC under Grants 62172016  and 61932001 for support. HW acknowledges NSFC under Grant 61972005 for support.

\bibliographystyle{named}
\bibliography{citation}

\begin{thebibliography}{}

\bibitem[\protect\citeauthoryear{Abdulkadiro{\u{g}}lu and
  S{\"o}nmez}{2003}]{abdulkadirouglu2003school}
Atila Abdulkadiro{\u{g}}lu and Tayfun S{\"o}nmez.
\newblock School choice: A mechanism design approach.
\newblock {\em American Economic Review}, 93(3):729--747, 2003.

\bibitem[\protect\citeauthoryear{Altunta{\c{s}} and
  Phan}{2022}]{altuntacs2022trading}
A{\c{c}}elya Altunta{\c{s}} and William Phan.
\newblock Trading probabilities along cycles.
\newblock {\em Journal of Mathematical Economics}, 100:102631, 2022.

\bibitem[\protect\citeauthoryear{Athanassoglou and
  Sethuraman}{2011}]{athanassoglou2011house}
Stergios Athanassoglou and Jay Sethuraman.
\newblock House allocation with fractional endowments.
\newblock {\em International Journal of Game Theory}, 40(3):481--513, 2011.

\bibitem[\protect\citeauthoryear{Aziz and Brandl}{2021}]{aziz2021efficient}
Haris Aziz and Florian Brandl.
\newblock {Efficient, fair, and incentive-compatible healthcare rationing}.
\newblock In {\em Proceedings of the 22nd ACM Conference on Economics and
  Computation}, pages 103--104, 2021.

\bibitem[\protect\citeauthoryear{Aziz and Brandl}{2022}]{aziz2022vigilant}
Haris Aziz and Florian Brandl.
\newblock The vigilant eating rule: A general approach for probabilistic
  economic design with constraints.
\newblock {\em Games and Economic Behavior}, 135:168--187, 2022.

\bibitem[\protect\citeauthoryear{Aziz}{2020}]{aziz2020simultaneously}
Haris Aziz.
\newblock Simultaneously achieving ex-ante and ex-post fairness.
\newblock In {\em International Conference on Web and Internet Economics},
  pages 341--355. Springer, 2020.

\bibitem[\protect\citeauthoryear{Babaioff \bgroup \em et al.\egroup
  }{2022}]{babaioff2022best}
Moshe Babaioff, Tomer Ezra, and Uriel Feige.
\newblock On best-of-both-worlds fair-share allocations.
\newblock In {\em International Conference on Web and Internet Economics},
  pages 237--255. Springer, 2022.

\bibitem[\protect\citeauthoryear{Balbuzanov}{2022}]{balbuzanov2022constrained}
Ivan Balbuzanov.
\newblock Constrained random matching.
\newblock {\em Journal of Economic Theory}, page 105472, 2022.

\bibitem[\protect\citeauthoryear{Biswas and Barman}{2018}]{biswas2018fair}
Arpita Biswas and Siddharth Barman.
\newblock Fair division under cardinality constraints.
\newblock In {\em Proceedings of the 27th International Joint Conference on
  Artificial Intelligence}, page 91–97. AAAI Press, 2018.

\bibitem[\protect\citeauthoryear{Biswas and Barman}{2019}]{biswas2019matroid}
Arpita Biswas and Siddharth Barman.
\newblock Matroid constrained fair allocation problem.
\newblock In {\em Proceedings of the AAAI Conference on Artificial
  Intelligence}, pages 9921--9922, 2019.

\bibitem[\protect\citeauthoryear{Bogomolnaia and
  Moulin}{2001}]{Bogomolnaia01:New}
Anna Bogomolnaia and Herv\'e Moulin.
\newblock {A new solution to the random assignment problem}.
\newblock {\em Journal of Economic Theory}, 100(2):295--328, 2001.

\bibitem[\protect\citeauthoryear{Budish \bgroup \em et al.\egroup
  }{2013}]{Budish2013:Designing}
Eric Budish, Yeon-Koo Che, Fuhito Kojima, and Paul Milgrom.
\newblock Designing random allocation mechanisms: Theory and applications.
\newblock {\em American Economic Review}, 103(2):585--623, 2013.

\bibitem[\protect\citeauthoryear{Budish \bgroup \em et al.\egroup
  }{2017}]{budish2017course}
Eric Budish, G{\'e}rard~P Cachon, Judd~B Kessler, and Abraham Othman.
\newblock Course match: A large-scale implementation of approximate competitive
  equilibrium from equal incomes for combinatorial allocation.
\newblock {\em Operations Research}, 65(2):314--336, 2017.

\bibitem[\protect\citeauthoryear{Budish}{2011}]{budish2011combinatorial}
Eric Budish.
\newblock The combinatorial assignment problem: Approximate competitive
  equilibrium from equal incomes.
\newblock {\em Journal of Political Economy}, 119(6):1061--1103, 2011.

\bibitem[\protect\citeauthoryear{Caragiannis \bgroup \em et al.\egroup
  }{2019}]{caragiannis2019unreasonable}
Ioannis Caragiannis, David Kurokawa, Herv{\'e} Moulin, Ariel~D Procaccia,
  Nisarg Shah, and Junxing Wang.
\newblock The unreasonable fairness of maximum nash welfare.
\newblock {\em ACM Transactions on Economics and Computation}, 7(3):1--32,
  2019.

\bibitem[\protect\citeauthoryear{Chen and S{\"o}nmez}{2006}]{chen2006school}
Yan Chen and Tayfun S{\"o}nmez.
\newblock School choice: An experimental study.
\newblock {\em Journal of Economic Theory}, 127(1):202--231, 2006.

\bibitem[\protect\citeauthoryear{Dror \bgroup \em et al.\egroup
  }{2021}]{dror2021fair}
Amitay Dror, Michal Feldman, and Erel Segal-Halevi.
\newblock On fair division under heterogeneous matroid constraints.
\newblock In {\em Proceedings of the AAAI Conference on Artificial
  Intelligence}, pages 5312--5320, 2021.

\bibitem[\protect\citeauthoryear{Dur \bgroup \em et al.\egroup
  }{2018}]{dur2018first}
Umut Dur, Timo Mennle, and Sven Seuken.
\newblock First-choice maximal and first-choice stable school choice
  mechanisms.
\newblock In {\em Proceedings of the 2018 ACM Conference on Economics and
  Computation}, pages 251--268, 2018.

\bibitem[\protect\citeauthoryear{Foley}{1966}]{foley1967resource}
Duncan~Karl Foley.
\newblock {\em Resource Allocation and the Public Sector}.
\newblock PhD thesis, Yale University, 1966.

\bibitem[\protect\citeauthoryear{Freeman \bgroup \em et al.\egroup
  }{2020}]{freeman2020best}
Rupert Freeman, Nisarg Shah, and Rohit Vaish.
\newblock Best of both worlds: Ex-ante and ex-post fairness in resource
  allocation.
\newblock In {\em Proceedings of the 21st ACM Conference on Economics and
  Computation}, pages 21--22, 2020.

\bibitem[\protect\citeauthoryear{Friedman}{1955}]{friedman1955role}
Milton Friedman.
\newblock The role of government in education.
\newblock In {\em Economics and the Public Interest}. Rutgers University Press,
  1955.

\bibitem[\protect\citeauthoryear{Friedman}{1962}]{friedman1962capitalism}
Milton Friedman.
\newblock Capitalism and freedom.
\newblock {\em Ethics}, 74(1), 1962.

\bibitem[\protect\citeauthoryear{Gamow and Stern}{1958}]{gamow1958puzzle}
George Gamow and Marvin Stern.
\newblock {\em {Puzzle-Math}}.
\newblock Viking Press, 1958.

\bibitem[\protect\citeauthoryear{Garg \bgroup \em et al.\egroup
  }{2010}]{garg2010assigning}
Naveen Garg, Telikepalli Kavitha, Amit Kumar, Kurt Mehlhorn, and Juli{\'a}n
  Mestre.
\newblock Assigning papers to referees.
\newblock {\em Algorithmica}, 58(1):119--136, 2010.

\bibitem[\protect\citeauthoryear{Ghodsi \bgroup \em et al.\egroup
  }{2011}]{Ghodsi11:Dominant}
Ali Ghodsi, Matei Zaharia, Benjamin Hindman, Andy Konwinski, Scott Shenker, and
  Ion Stoica.
\newblock {Dominant resource fairness: Fair allocation of multiple resource
  types}.
\newblock In {\em {Proceedings of the 8th USENIX Conference on Networked
  Systems Design and Implementation}}, pages 323--336, 2011.

\bibitem[\protect\citeauthoryear{Grandl \bgroup \em et al.\egroup
  }{2014}]{Grandl15:Multi}
Robert Grandl, Ganesh Ananthanarayanan, Srikanth Kandula, Sriram Rao, and
  Aditya Akella.
\newblock Multi-resource packing for cluster schedulers.
\newblock In {\em Proceedings of the 2014 ACM Conference on SIGCOMM}, pages
  455--466, 2014.

\bibitem[\protect\citeauthoryear{Guo \bgroup \em et al.\egroup
  }{2023}]{guo2021favoring}
Xiaoxi Guo, Sujoy Sikdar, Lirong Xia, Yongzhi Cao, and Hanpin Wang.
\newblock Favoring eagerness for remaining items: Designing efficient, fair,
  and strategyproof mechanisms.
\newblock {\em Journal of Artificial Intelligence Research}, 76:287--339, 2023.

\bibitem[\protect\citeauthoryear{Hoefer \bgroup \em et al.\egroup
  }{2022}]{hoefer2022best}
Martin Hoefer, Marco Schmalhofer, and Giovanna Varricchio.
\newblock Best of both worlds: Agents with entitlements.
\newblock {\em arXiv preprint arXiv:2209.03908}, 2022.

\bibitem[\protect\citeauthoryear{Hosseini and
  Larson}{2019}]{hosseini2019multiple}
Hadi Hosseini and Kate Larson.
\newblock Multiple assignment problems under lexicographic preferences.
\newblock In {\em Proceedings of the 18th International Conference on
  Autonomous Agents and MultiAgent Systems}, pages 837--845, 2019.

\bibitem[\protect\citeauthoryear{Hosseini \bgroup \em et al.\egroup
  }{2021}]{hosseini2021fair}
Hadi Hosseini, Sujoy Sikdar, Rohit Vaish, and Lirong Xia.
\newblock Fair and efficient allocations under lexicographic preferences.
\newblock In {\em Proceedings of the AAAI Conference on Artificial
  Intelligence}, pages 5472--5480, 2021.

\bibitem[\protect\citeauthoryear{Irving \bgroup \em et al.\egroup
  }{2006}]{irving2006rank}
Robert~W. Irving, Telikepalli Kavitha, Kurt Mehlhorn, Dimitrios Michail, and
  Katarzyna~E. Paluch.
\newblock Rank-maximal matchings.
\newblock {\em ACM Transactions on Algorithms}, 2(4):602–610, 2006.

\bibitem[\protect\citeauthoryear{Kawase \bgroup \em et al.\egroup
  }{2020}]{kawase2020subgame}
Yasushi Kawase, Yutaro Yamaguchi, and Yu~Yokoi.
\newblock Subgame perfect equilibria of sequential matching games.
\newblock {\em ACM Transactions on Economics and Computation}, 7(4):1--30,
  2020.

\bibitem[\protect\citeauthoryear{Kirkpatrick \bgroup \em et al.\egroup
  }{2020}]{kirkpatrick2020scarce}
James~N Kirkpatrick, Sarah~C Hull, Savitri Fedson, Brendan Mullen, and Sarah~J
  Goodlin.
\newblock {Scarce-resource allocation and patient triage during the COVID-19
  pandemic}.
\newblock {\em Journal of the American College of Cardiology}, 76(1):85--92,
  2020.

\bibitem[\protect\citeauthoryear{Kojima and
  {\"{U}}nver}{2014}]{Kojima2014:Boston}
Fuhito Kojima and M.~Utku {\"{U}}nver.
\newblock {The ``Boston'' school-choice mechanism: an axiomatic approach}.
\newblock {\em Economic Theory}, 55(3):515--544, 2014.

\bibitem[\protect\citeauthoryear{Kojima}{2009}]{kojima2009random}
Fuhito Kojima.
\newblock Random assignment of multiple indivisible objects.
\newblock {\em Mathematical Social Sciences}, 57(1):134--142, 2009.

\bibitem[\protect\citeauthoryear{Li}{2020}]{li2020ties}
Mengling Li.
\newblock Ties matter: Improving efficiency in course allocation by allowing
  ties.
\newblock {\em Journal of Economic Behavior \& Organization}, 178:354--384,
  2020.

\bibitem[\protect\citeauthoryear{Manlove}{2013}]{manlove2013algorithmics}
David Manlove.
\newblock {\em Algorithmics of Matching under Preferences}.
\newblock World Scientific, 2013.

\bibitem[\protect\citeauthoryear{Moulin}{2004}]{Moul04}
H.~Moulin.
\newblock {\em Fair Division and Collective Welfare}.
\newblock MIT Press, 2004.

\bibitem[\protect\citeauthoryear{Pathak \bgroup \em et al.\egroup
  }{2021}]{pathak2021fair}
Parag~A Pathak, Tayfun S{\"o}nmez, M~Utku {\"U}nver, and M~Bumin Yenmez.
\newblock {Fair allocation of vaccines, ventilators and antiviral treatments:
  Leaving no ethical value behind in health care rationing}.
\newblock In {\em Proceedings of the 22nd ACM Conference on Economics and
  Computation}, pages 785--786, 2021.

\bibitem[\protect\citeauthoryear{Ramezanian and
  Feizi}{2021}]{Ramezanian2021:Ex-post}
Rasoul Ramezanian and Mehdi Feizi.
\newblock {Ex-post favoring ranks: A fairness notion for the random assignment
  problem}.
\newblock {\em Review of Economic Design}, 25:157–176, 2021.

\bibitem[\protect\citeauthoryear{Sayedahmed and
  others}{2022}]{sayedahmed2022centralized}
Dilek Sayedahmed et~al.
\newblock Centralized refugee matching mechanisms with hierarchical priority
  classes.
\newblock {\em The Journal of Mechanism and Institution Design}, 7(1):71--111,
  2022.

\bibitem[\protect\citeauthoryear{Wang \bgroup \em et al.\egroup
  }{2023}]{wang2023multi}
Haibin Wang, Sujoy Sikdar, Xiaoxi Guo, Lirong Xia, Yongzhi Cao, and Hanpin
  Wang.
\newblock Multi resource allocation with partial preferences.
\newblock {\em Artificial Intelligence}, 314:103824, 2023.

\bibitem[\protect\citeauthoryear{Y{\i}lmaz}{2010}]{yilmaz2010probabilistic}
{\"O}zg{\"u}r Y{\i}lmaz.
\newblock The probabilistic serial mechanism with private endowments.
\newblock {\em Games and Economic Behavior}, 69(2):475--491, 2010.

\end{thebibliography}

\clearpage

\appendix

\section*{Appendix}


\begin{restatable}{prop}{proprsdpnefo}{}\label{prop:rsdpnefo}
		RSDQ does not satisfy \sdefo{}.
	\end{restatable}
    \begin{proof}
        We use the instance  with the following preferences profile $R$:
        \begin{equation*}
            \begin{split}
            \succ_{1,2}\text{: }& a\succ b\succ c\succ d.\\
            \end{split}
        \end{equation*}
        To guarantee \sdefo{}, each agent should be allocated two items, i.e., each agent's quota is $2$.
        It means that in the execution of RSDQ, agents in turn pick two items according to the given priority order $\impord{}$.
        However, $1\impord{} 2$ results in $1\gets(a,b),2\gets(c,d)$ and $2\impord{} 1$ results in $1\gets(c,d),2\gets(a,b)$, both of which violates \sdefo{}.
        It means that RSDQ does not satisfy \sdefo{} even for two agents with the quota $2$.
    \end{proof}


    \lemopt*
    \begin{proof}
        It is trivial that a \opt{} assignment $A$ should have no cycle for the relation.
        We show  $A$ is also a \opt{} assignment if it does not admit a cycle.

        Assume by contradiction that $A$ is lexicographically dominated by another \das{} $A'$.
        Let $N'=\{j\in N\mid A'(j)\ld{j}A(j)\}$ be the set of agents with better allocations in $A'$, which means that $A(k)=A'(k)$ for any other agent $k\in N\setminus N'$.
        For $j\in N'$, there exists item $o_j$ such that $o_j\in A'(j)$ and $o_j\notin A(j)$.
        It also means every agent $j$ in $N'$ must obtain $o_j$ from another agent who is also in $N'$.
        Then we can build up the following sequence $Seq$ as long as possible $j_1, j_2, \dots, j_h \in N'$ such that $j_1$ takes $o_h$ from $j_h$ and $j_{h'-1}$ takes $o_{j_{h'-1}}$ from $j_{h'}$ where $1<h'<h$.
        Since $N'$ is a finite set, the sequence is also finite, and therefore agent $j_h$ can only get $o_{j_h}$ from an agent who has already existed in $Seq$, which forms a cycle that contradicts the fact that $A$ does not admit a cycle.
    \end{proof}


    \clmmfhcr*
    \begin{proof}
    Consider an arbitrary agent $j$, and let $o=A^c(j)$ be the item allocated to agent $j$ in round $c$ of $\mam{}(N,M,R)$ through the execution of $\am{}(N,M^c,R)$. We will prove that for any item $o'\in M^{c+1}$, it holds that $o\succ_j o'$. Suppose for the sake of contradiction that $o'\succ_j o$. 
    
    Notice that by the construction of $M^c$ and $M^{c+1}$, it must hold that both $o$ and $o'$ must be in $M^c$. Now, since $A^c$ satisfies \fhcr{} for the matching instance $(N,M^c)$ under the preference profile $R$ (\Cref{clm:thm:mam:1}), we have by the definition of \fhcr{} that $o\in\tps{A^c}{r^*}$ for some $r^*$,
    \begin{equation}\label{eq:clm:mfhcr:1}
        o=\tp{j}{M^c\setminus{\bigcup_{r'<r^*}\tps{A^c}{r'}}}
    \end{equation}
    i.e. $o$ is agent $j$'s top item in $\tps{A^c}{r}$. Recall here that $\tps{A^c}{r}$ is the set of items in $M^c$ that are most preferred by some agent given the allocations of items in $\bigcup_{r'<r}\tps{A^c}{r'}$, whose owners cannot receive any more items in the matching instance $(N,M^c)$ in round $c$ within the execution of $\mam{}(N,M,R)$, i.e., $\tps{A^c}{r'}=\{o\in M^c| o=\tp{k}{M^c}\bigcup_{r'<r}\tps{A^c}{r'} \text{ for some agent } k\in N \text{ such that } A^c(k)\notin \bigcup_{r'<r}\tps{A^c}{r'}\}$. Also recall that since $A^c$ satisfies \fhcr{}, the set of items allocated within round $c$ of $\mam{}(N,M,R)$ is $M^c\setminus M^{c+1}=\bigcup_{r\ge0}\tps{A^c}{r}$.

    It follows that since $o'\in M^{c+1}$, $o'\notin\tps{A^c}{r}$ for any value of $r\ge 0$. Specifically, this implies that $o'\notin\tps{A^c}{r^*}$ and $o'\in M^c\setminus \bigcup_{r<r^*}\tps{A^c}{r}$. It follows that $o\succ_j o'$, a contradiction to our assumption, completing the proof.
    \end{proof}

    \clmgam*
    \begin{proof}
        Let $p=\sum_{i=1}^sp_{i}$ and $q=\sum_{i=1}^sq_{i}$. For any pair of $p_i$ and $q_i$ with $p_i \neq q_i$, there exists an item $o_i$ such that $p_i({o_i})>q_i({o_i})$ and $p_i({o'})=q_i({o'})$ for any $o'\succ o_i$.
            Let $M'$ be the set of all such items.
        If $M'=\emptyset$, then $p=q$ and the claim trivially holds.
        When $M'\neq\emptyset$, let $o_i$ be the one ranked highest in $M'$ without loss of generality.
        We prove that $p({o_i})>q({o_i})$ and $p({o'})=q({o'})$ for any $o'\succ o_i$ and therefore $p\ld{} q$.
        Assume by contradiction that there exists $o\succ_j o_i$ with $p({o})<q({o})$.
        It means that there exists a pair  of $p_h$ and $q_h$ such that $p_h({j,o})<q_h({k,o})$.
        By the condition that $p_h\ld{}q_h$, we have that there exists another item $o_h\succ o\succ o_i$ with $p_h({o_h})>q_h({o_h})$ and $p_h({o'})=q_h({o'})$ for any $o'\succ_j o_h$, a contradiction to the selection of $o_i$.
        \end{proof}


        \begin{definition}[\bf trivial extension for \fhr{}]\label{dfn:tfhr}
		Given a preference profile $R$, an \das{} $A$ satisfies {\bf \ffhr{} (\fhr{})} if for any agents $j$ and $k\in N$ and any items $o_j\in A(j)$ and $o_k\in A(k)$, $\rk{j}{o_j}\le\rk{k}{o_j}$ or $\rk{k}{o_k}<\rk{k}{o_j}$.
	\end{definition}

    \begin{restatable}{prop}{propimp_fhr}{}\label{prop:imp_fhr}
		An allocation mechanism cannot satisfy both \fhr{} and \sdefo{}.
	\end{restatable}
    
    \begin{proof}
        Given an instance with $4$ agents, $8$ items, and the following preference profile:
        
        \begin{equation*}
            \begin{split}
            \succ_{1-3}\text{: }& a\succ b\succ c\succ d\succ\text{others,}\\
            \succ_4\text{: }& c\succ_4 d\succ_4\text{ others.}
            \end{split}
        \end{equation*}

        Let $A$ be any \das{} satisfying \fhr{}.
        According to \fhr{}, it is easy to see that agent $4$ is allocated items $c$ and $d$ in $A$.
        Meanwhile, items $a$ and $b$ are allocated to agents $1-3$ in $A$, which means that at least one agent among them does not get either $a$ or $b$.
        Without loss of generality, let this agent be $1$.
        Then we see that $A(1)\sd{1}A(4)\setminus{}\{o\}$ does not hold no matter $o$ is either $c$ or $d$, which means that $A$ can never satisfy \sdefo{}.
    \end{proof}

    \begin{algorithm}[htb]
		\begin{algorithmic}[1]
			\State {\bf Input:} An assignment problem $(N,M)$, random assignments $P^1,\dots,P^{\lceil{m/n}\rceil}$.
            \State $M'\gets M\cup\{{\bf nil}\}$. $Q\gets 0^{(n*\lceil{m/n}\rceil)\times m}$.
            \State For each agent $j\in N$, create subagents $j^1,\dots,j^{\lceil{m/n}\rceil}$.
            \For{$j\in N$}
                \For{$o\in M$}
                    \State $Q(j^c)\gets P^c(j)$.
                \EndFor
                \If{$\sum_{o\in M} P^{\lceil{m/n}\rceil}(j,o)<1$}
                    \State $Q(j^{\lceil{m/n}\rceil},{\bf nil})\gets 1-\sum_{o\in M} P^{\lceil{m/n}\rceil}(j,o)$.
                \EndIf
            \EndFor
            \State Draw a $(n*\lceil{m/n}\rceil)\times \lvert M'\rvert$ \das{} $\hat{A}$ according to the probabilistic distribution represented by bistochastic random assignment $Q$.
            \State Build the $n\times m$ \das{} $A$ such that for $j\in N$ and $o\in M$, $A(j,o)\gets \sum_{c=1}^{\lceil{m/n}\rceil}\hat{A}(j^c,o)$.
            \State \Return  $A$.
		\end{algorithmic}
		\caption{\label{alg:generate} Generating an \das{}}
	\end{algorithm}

 \clmgpbm*

 \begin{proof}
        We prove that the relation in~\Cref{lem:sdopt} is acyclic in $P^c$.
        Assume that there is such a cycle, and there exists an agent $j\in N$ with $o\succ_j o'$ and $P^c(j,o')>0$.
        We observe that in round $c$ of \mpbm{} that generates $P^c$, agent $j$ consumes item $o'$ at consumption round $r'=\rk{j}{o'}$.
        We also have that item $o$ is exhausted by the end of consumption round $r=\rk{j}{o}<r'$.
        Otherwise, it follows that agent $j$ stops consuming item $o$ at consumption round $r$ because $\sum_{o''\in\ucs(\succ_j,o)}P^c(j,o'')=1$, which contradicts $P^c(j,o')>0$.
        Then we have that for any pair of such $o$ and $o'$, item $o$ is consumed to exhaustion at an earlier consumption round than $o'$.
        Since $o,o'$ are in the cycle, we also have that item $o'$ is consumed to exhaustion at an earlier than $o$, which is a contradiction.
    \end{proof}

\clmthmmpbm*

\begin{proof}
            \noindent{\bf\boldmath  (\romannumeral1) $A^c$ satisfies \fhr{}.} 
            Suppose for the sake of contradiction that there exists a pair of agents $j,k\in N$ such that $o=A^c(j)\succ_k A^c(k)=o'$ and $r=\rk{j}{o}>\rk{k}{o}=r'$, violating \fhr{}.
            It follows that $P^c({j,o})>0$ and $P^c({k,o'})>0$.
            Then, in round $c$ of \mpbm{} in which $P^c$ is computed, agent $j$ must consume $o$ at consumption round $r$, which means $s(o)>0$ at the beginning of consumption round $r$, i.e., item $o$ is not exhausted and therefore available to consumer in the earlier consumption round $r'<r$.
            By construction, agent $k$ must consume $o$ at consumption round $r'$ and item $o$ is not exhausted, which implies that $\sum_{o''\in\ucs(\succ_k,o)}P^c({k,o''})=1$ and therefore $P^c({k,o'})=0$ since $o\succ_k o'$, which is a contradiction.
            
            \vspace{0.5em}
            \noindent{\bf\boldmath  (\romannumeral2) $A^c(j)\succ_j A^{c+1}(k)$.}
            Assume for the sake of contradiction that $A^{c+1}(k)=o'\succ_j o=A^c(j)$, which implies that $r'=\rk{j}{o'}<\rk{j}{o}$.
            It follows that $P^{c}({j,o})>0$ and $P^{c+1}({k,o'})>0$, which means $s(o')>0$ at the beginning of round $c+1$ of \mpbm{} used to compute $P^{c+1}$.
            At round $c$ of \mpbm{}, agent $j$ is able to consume $o'$ in consumption round $r'$, during which $o'$ is not consumed to exhaustion.
            It follows that $\sum_{o''\in\ucs(\succ_j,o')}P^c=1$, which implies that $P^c(j,o)=0$ since $o'\succ_j o$, a contradiction.
        \end{proof}

\end{document}